\documentclass[a4paper]{article}

\usepackage[T1]{fontenc}
\usepackage{lmodern}
\usepackage{geometry}
\usepackage{amsmath,amsthm,amssymb}
\usepackage{graphicx}
\usepackage{microtype}  

\usepackage{titlesec}
\titleformat*{\section}{\large\bfseries}
\titleformat*{\subsection}{\bfseries}

\makeatletter
  \def\@fnsymbol#1{\ensuremath{\ifcase#1\or *\or ** \or \ddagger\or
  \mathsection\or \mathparagraph\or \|\or **\or \ddagger\ddagger \else\@ctrerr\fi}}
\makeatother

\newtheorem{lemma}{Lemma}[section]
\newtheorem{claim}{Claim}[section]
\newtheorem{corollary}{Corollary}[section]

\title{Reachability and Shortest Paths in the Broadcast \\CONGEST Model{\thanks{\,A preliminary version of this paper appears in the \textit{Proceedings of the 33rd International Symposium on Distributed Computing (DISC'19)}, pages 11:1--11:13, 2019. The current version contains some improvements. In particular, it improves the approximate APSP algorithm.}}}
\author{Shiri Chechik\thanks{\,Tel Aviv University, Israel; emails: {\tt shiri.chechik@gmail.com} and {\tt doron.muk@gmail.com}.} \and Doron Mukhtar{$^{**}$}}
\date{}

\usepackage{xcolor}
\usepackage{hyperref} \hypersetup{
    colorlinks = true, allcolors = black, citecolor = {blue!60!black}, linkcolor = {blue!60!black}
}

\begin{document}

\maketitle\thispagestyle{empty}
 
\begin{abstract} In this paper we study the time complexity of the single-source reachability problem and the single-source shortest path problem for directed unweighted graphs in the Broadcast CONGEST model. We focus on the case where the diameter $D$ of the underlying network is constant. 

We show that for the case where $D = 1$ there is, quite surprisingly, a very simple algorithm that solves the reachability problem in $1$(!) round. In contrast, for networks with $D = 2$, we show that any distributed algorithm (possibly randomized) for this problem requires $\Omega(\sqrt{n/ \log{n}}\,)$ rounds. Our results therefore completely resolve (up to a small polylogarithmic factor) the complexity of the single-source reachability problem for a wide range of diameters.

Furthermore, we show that when $D = 1$, it is even possible to get a $3$\,--\,approximation for the all-pairs shortest path problem (for directed unweighted graphs) in just $2$ rounds. We also prove a stronger lower bound of $\Omega(\sqrt{n}\,)$ for the single-source shortest path problem for unweighted directed graphs that holds even when the diameter of the underlying network is $2$. As far as we know this is the first lower bound that achieves $\Omega(\sqrt{n}\,)$ for this problem.\end{abstract}

\newpage \setcounter{page}{1}

\section{Introduction}

Reachability and shortest path are two of the most fundamental problems in graph algorithms. In this paper, we study the single-source reachability (SSR) problem and the single-source shortest path (SSSP) problem in the Broadcast CONGEST model of distributed computing.

The CONGEST model \cite{Peleg00b} is one of the most studied message-passing models in the field of distributed computing. In this model, a synchronized $n$-vertex communication network is modeled by an undirected graph $N$ whose vertices correspond to the processors in this network and whose edges correspond to the communication links between them. Each vertex has a unique $O(\log{n})$-bit identifier initially known only to itself and its neighbors in $N$. The vertices communicate in discrete rounds, where in each round each vertex receives the messages that were previously sent to it, performs some unbounded local computation and then sends messages of $O(\log{n})$ bits to all or some of its neighbors. The vertices work together on some common task (such as computing distances in the network) and the complexity is measured by the number of communication rounds needed to complete this task. The Broadcast CONGEST model is a more restrictive variant of the CONGEST model where every vertex has to send (broadcast) the same message to all of its neighbors in each round.

In this paper we focus on directed and unweighted graphs. In the SSR problem, we are asked to identify all the vertices in a given graph $G$ for which there is a directed path from some designated vertex $s$ called the source. In the SSSP problem, we are further asked to compute for each such vertex its distance (the number of edges in a shortest path) from the source $s$. In the CONGEST model as well as in other similar message-passing models, we assume that the communication network $N$ is identical to the underlying graph of $G$ (where $G$ is the input graph for the SSR\textbackslash SSSP problem). We also assume that the communication between the vertices is bi-directional (regardless of the directions of the edges in $G$). Initially, each vertex in the network knows whether it is the source or not, and it also knows its set of incoming and outgoing edges in $G$. In the distributed SSR problem, each vertex has to determine whether it is reachable from the source or not, and in the distributed SSSP problem, each vertex has to determine its distance from the source.

\subparagraph{Related Work} 

Distance computation problems (such as the SSSP problem) have been widely studied in many models of distributed computing. It is not hard to see that in many synchronous message-passing models, problems such as SSR and SSSP require $\Omega(D)$ rounds (where $D$ is the diameter of the underlying network). While this lower bound can be easily matched when messages of unbounded size are allowed, the situation for models that require the messages to be of bounded size is far more involved.

In the CONGEST model, the directed SSR problem has been studied in several papers \cite{Nan14,GR15,Arun19}. The latest \cite{Arun19} shows that it is possible to solve this problem in $\tilde{O}(n^{1/2} + n^{1/3+o(1)}D^{2/3})$ rounds with high probability. Many variants of the SSSP problem (directed\textbackslash undirected, exact\textbackslash approximate etc.) were studied over the years (see, e.g., \cite{Nan14, Hen16, Beck17, Elkin17, GhaLi18, SFDN18}). In particular, for directed and weighted graphs, there is a randomized algorithm that solves the SSSP problem in $\tilde{O}(\sqrt{nD}\,)$ rounds \cite{SFDN18}. We note that many of the above mentioned algorithms (such as \cite{GR15, SFDN18}) actually work in the more restrictive Broadcast CONGEST model. Regarding lower bounds, Das Sarma et al. \cite{DasSa11} showed that in the CONGEST model the time complexity of any (possibly randomized) algorithm for the directed single-source reachability problem is $\Omega(\sqrt{n/ \log{n}}\,)$. However, this lower bound was shown only for graphs of underlying diameter $\Omega(n^\delta)$ for some $0 < \delta < 1/2$. For smaller diameters, similar but weaker lower bounds were shown e.g. $\Omega(\sqrt{n}/ \log{n})$ for graphs of underlying diameter $\Theta(\log n)$. The smallest constant diameter for which a non-trivial lower bound is known is 3 where it was shown to require $\Omega((n/\log{n})^{1/4})$ rounds \cite{DasSa11}.

For the related all-pairs shortest path (APSP) problem, many algorithms with near-optimal complexities for the approximate version of this problem and for the case of unweighted graphs were developed over the years (e.g., \cite{Holz12, Len15, Len13, Nan14}). Recently, many algorithms with improved complexities for the case of weighted graphs were devised \cite{Elkin17, CDT17, AgaU18} culminating with the $\tilde{O}(n)$-time randomized algorithm of \cite{AD19}.

\subparagraph{Our Results}

Many typical real-world networks usually have a relatively small diameter (as argued in e.g. \cite{Nan14}) and, in many cases, a simple topology as well. It is thus of particular interest to understand the complexity of many optimization problems in such settings. Motivated by this, we study in this paper the time complexity of the SSR problem and the SSSP problem (for directed unweighted graphs) in the Broadcast CONGEST model for networks of constant diameter. Specifically, we show that even for networks of diameter $2$ with a very simple topology, any distributed algorithm (possibly randomized) for the SSR problem requires $\Omega(\sqrt{n/ \log{n}}\,)$ rounds. In contrast, we show that quite surprisingly for networks of diameter $1$, this problem (or even the more general all-pairs reachability problem) can be solved deterministically in $1$ round. Moreover, we show that for networks of diameter $1$ one can compute in $2$ rounds a $3$-approximation for the SSSP\textbackslash APSP problem.

The algorithm for the approximate APSP problem (resp. for the all-pairs reachability problem) allows each vertex to compute a $3$-approximation for the distance between every pair of vertices in the graph (resp. determine reachability for every such pair) and not only for the pairs to which it belongs. We note that if one can compute a $(2-\epsilon)$-approximation for the APSP problem (for some $1 \ge \epsilon > 0$) such that there is some vertex $v$ that knows the computed estimation for every pair of vertices, then this vertex can recover the whole graph. This means that $v$ must receive in this case $\Omega(n^2)$ bits of information from its neighbors (simply because there are $2^{\,\Omega(n^2)}$ possible graphs on these vertices), but in each round, $v$ can get at most $O(n\log{n})$ bits from its neighbors and so $\Omega(n/\log{n})$ rounds are required for solving this problem.

Our results show a large gap between networks of diameter $1$ and $2$. As upper bounds of $\tilde{O}(\sqrt{n}\,)$ are already known for the SSR problem when the underlying network has a constant or polylogarithmic diameter (e.g., \cite{GR15,SFDN18}), we completely resolve (up to polylogarithmic factors) the complexity of the SSR problem when the diameter of the underlying network is in that range. Our algorithms are very simple (we see this as a plus and not a minus). In addition, we show a stronger lower bound of $\Omega(\sqrt{n}\,)$ for the SSSP problem for unweighted directed graphs in the Broadcast CONGEST model that holds even when the diameter of the underlying network is 2. As far as we know this is the first lower bound that achieves $\Omega(\sqrt{n}\,)$ for this problem. 

\subparagraph{Further Related Work}

A closely related model to the CONGEST when the underlying communication network has diameter 1 is the Congested Clique model. The Congested Clique model is a synchronous message-passing model in which the underlying communication network is the complete graph on $n$ vertices but the graph $G$ on which the solution needs to be obtained can be an arbitrary graph on $n$ vertices (that is, each vertex initially knows its neighbors in $G$ and can exchange messages of size $O(\log{n})$ with any vertex in the graph even if they are not adjacent in $G$).

Censor-Hillel et al. \cite{CensorH15} adapted parallel matrix multiplication algorithms to this model. Using these algorithms, they obtained better algorithms for subgraph detection and distance computation. In particular, they showed a $\tilde{O}(n^{1/3})$-round algorithm for solving the APSP problem for weighted directed graphs, and even more efficient algorithms for unweighted undirected graphs or distance approximation. Recently, it was shown \cite{CensorH19} that the SSSP problem for weighted undirected graphs can be solved in $\tilde{O}(n^{1/6})$ rounds.

We note that for problems such as SSSP or APSP (for weighted graphs) the Congested Clique model is actually a special case of the CONGEST model when the diameter of the underlying network is 1. To see this, note that one can always transform the input graph $G$ into a complete graph by adding edges of very large weight. Therefore, either one can show a constant upper bound for the weighted SSSP problem in the Congested Clique model which will be quite a breakthrough or our upper bound shows a separation between the SSR problem and the SSSP problem for directed weighted graphs of underlying diameter 1 (and even between the all-pairs reachability problem and the SSSP problem).

\section{Preliminaries}

In the following, we assume that all directed graphs are simple (i.e., they do not contain self-loops or multiple edges, but they may contain anti-parallel edges). For a graph $H$, we respectively denote by $V(H)$ and $E(H)$ its vertex set and edge set. The out-degree and in-degree of a vertex $v$ in a directed graph $H$ are denoted by $d_{out}(v)$ and $d_{in}(v)$, respectively. For a directed graph $H$ and a vertex $v$ in $H$, we denote by $N_{out}(v)$ its set of outgoing neighbors, and by $N_{in}(v)$ its set of ingoing neighbors. Given a directed graph $H = (V,E)$ and a set $A\subseteq V$, we denote by $A^c$ the set $V \setminus A$. The underlying diameter of a directed graph $H$ is defined to be the diameter of its underlying graph. For a graph $H$ and two vertices $v$ and $u$ in $V(H)$, we denote by $d(v,u,H)$ the distance from $v$ to $u$ in $H$. All logarithms in this paper are of base $2$.

The rest of the paper is organized as follows. In section \ref{sec:reachability-1}, we show an algorithm that solves the all-pairs reachability problem in one round for networks of diameter 1. In section \ref{sec:apsp}, we show that in two rounds one can compute an approximation for the APSP problem (also for networks of diameter 1). In section \ref{sec:lower_bounds}, we prove lower bounds for computing reachability and distances in networks of diameter 2.

\section{All-Pairs Reachability for Networks of Diameter 1} \label{sec:reachability-1}

In this section we show that when the diameter of the underlying network is 1, the directed single-source reachability problem can be solved in $O(1)$ rounds in the Broadcast CONGEST model. In fact, we show that it can be solved in a single round.
Furthermore, our algorithm can solve the much more general problem of all-pairs reachability (again in a single round).
The algorithm is extremely simple. Every vertex simply sends its in-degree and out-degree to all of its neighbors in the underlying network, and then, by using this information only, each vertex can determine (by a simple computation) which vertex is reachable from which. This requires messages of at most $2\lceil\log_2{n}\rceil$ bits where $n$ is the number of vertices in the network (moreover, if there are no anti-parallel edges then, as the underlying diameter is 1, the in-degree plus out-degree of every vertex is exactly $n-1$ and therefore it is enough to send only the in-degree and so $\lceil\log_2{n}\rceil$ bits are enough).

The next lemma shows that when the underlying diameter of some directed graph $H$ is 1, we can determine if $E(H) \cap (A \times A^c) = \varnothing$ by using the in and out degrees of the vertices in $A$, for every subset of vertices $A \subseteq V(H)$.

\begin{lemma} For every directed graph $H = (V,E)$ with underlying diameter $1$ and every set $A \subseteq V$, we have $\sum_{v \in A}(d_{in}(v) - d_{out}(v)) = |A^c \times A|$ if and only if $E\cap (A \times A^c) = \varnothing$.\label{lem:closed_set}\end{lemma}

\begin{proof} Let $H = (V,E)$ be a directed graph with underlying diameter $1$ and let $A$ be some subset of $V$. We have $\sum_{v \in A}{d_{out}(v)} = |E \cap (A \times V)| = |E \cap (A \times A)| + |E \cap (A \times A^c)|$ and similarly $\sum_{v \in A}{d_{in}(v)} = |E \cap (V \times A)| = |E \cap (A \times A)| + |E \cap (A^c \times A)|$. It follows that \begin{equation}\quad\quad\\\\ \textstyle{\sum}_{v\in A}{\left(d_{in}(v) - d_{out}(v) \right)} = |E\cap (A^c \times A)| - |E \cap (A \times A^c)|\label{eq1}\end{equation} Now, for showing the first direction, assume that $\sum_{v \in A}(d_{in}(v) - d_{out}(v)) = |A^c \times A|$. By equation (\ref{eq1}), we have $|A^c \times A| = |E\cap (A^c \times A)| - |E \cap (A \times A^c)|$. As $|E \cap (A^c \times A)| \le |A^c \times A|$, we get that $|E \cap (A \times A^c)| \le 0$ and so $E \cap (A \times A^c) = \varnothing$.

For the second direction, assume that $E \cap (A \times A^c) = \varnothing$. Since in addition $H$ has underlying diameter $1$, every vertex in $A^c$ must have an outgoing edge to every vertex in $A$ and so $E \cap (A^c \times A) = A^c \times A$. It follows, by equation (\ref{eq1}), that $\sum_{v\in A}{\left(d_{in}(v) - d_{out}(v) \right)} = |A^c \times A|$.\end{proof}

The next lemma shows that when the underlying diameter of some directed graph $H$ is 1, the in and out degrees of all the vertices in $H$ are enough to determine which vertices are reachable from any given vertex in $H$.

\begin{lemma} For every directed graph $H$ with underlying diameter $1$, every ordering $(v_1,...,v_n)$ of its vertices  such that $d_{out}(v_1) \le ... \le d_{out}(v_n)$ and every $i \in \{1,...,n\}$, there exists an index $k \in \{i,...,n\}$ such that the set of reachable vertices from $v_i$ in $H$ is equal to $\{v_1,...,v_k\}$. Moreover, $k$ is the minimal index in $\{i,...,n\}$ for which $(n-k)k = \sum_{j=1}^k(d_{in}(v_j) - d_{out}(v_j))$.\label{lemma:reachability}\end{lemma}

\begin{proof} Let $H = (V,E)$ be a directed graph with underlying diameter $1$, let $(v_1,...,v_n)$ be an ordering of its vertices such that $d_{out}(v_1) \le ... \le d_{out}(v_n)$ and let $i \in \{1,...,n\}$. Let $A$ be the set of all the reachable vertices from $v_i$ in $H$, and note that we must have $E \cap (A \times A^c) = \varnothing$ (as otherwise $v_i$ can reach a vertex from $A^c$ which is of course contradiction to the definition of $A$ and $A^c$).

Let $k$ be the highest index in $\{i,...,n\}$ for which $v_k \in A$ (such an index must exist as $v_i \in A$). Clearly, we have $A \subseteq \{v_1,...,v_k\}$. We claim that we must also have $\{v_1,...,v_k\} \subseteq A$. Since $v_k \in A$ and $E \cap (A \times A^c) = \varnothing$, the set $A$ must contain at least $d_{out}(v_k) + 1$ vertices (the vertex $v_k$ and its $d_{out}(v_k)$ outgoing neighbors).
It also follows that every vertex in $A^c$ must have out-degree at least $d_{out}(v_k) + 1$.
To see this, note that every vertex in $A^c$ must have an outgoing edge to every vertex in $A$ (as $E \cap (A \times A^c) = \varnothing$ and the underlying diameter of $H$ is $1$). Therefore, it must be that $v_j \in A$ for all $j \in \{1,...,k\}$ as $d_{out}(v_j) \le d_{out}(v_k)$ for every such $j$. We conclude that $A = \{v_1,...,v_k\}$.

Now, as $E \cap (A \times A^c) = \varnothing$ we get from Lemma \ref{lem:closed_set} that $(n-k)k = \sum_{j=1}^k(d_{in}(v_j) - d_{out}(v_j))$. We are left to show that $k$ is the minimal index in $\{i,...,n\}$ with this property. Assume towards a contradiction that there exists $m \in \{i,...,n\}$ such that $m < k$ and $\sum_{j=1}^m(d_{in}(v_j) - d_{\text{out}}(v_j)) = (n-m)m$. Let $B = \{v_1,...,v_m\}$. By Lemma \ref{lem:closed_set} we get that $E \cap (B \times B^c) = \varnothing$, and, in particular, that $v_k$ is not reachable from $v_i$ (as $v_i \in B$ and $v_k \not\in B$) which is a contradiction.\end{proof}

Lemma \ref{lemma:reachability} can be easily turned into an algorithm that solves the all-pairs reachability problem in one round (when the diameter of the underlying network is 1) as follows. Each vertex $v$ in the graph starts by broadcasting the values of $d_{in}(v)$ and $d_{out}(v)$. After receiving the messages, $v$ sorts the vertices in non-decreasing order of their out-degree. Let $(v_1,...,v_n)$ be that ordering. It then finds for every $i \in \{1,...,n\}$ the minimal index $k_i \in \{i,...,n\}$ such that $(n-k_i)k_i = \sum_{j=1}^{k_i}(d_{in}(v_j) - d_{out}(v_j))$ and deduces by Lemma \ref{lemma:reachability} that the set of reachable vertices from $v_i$ is $\{v_1,...,v_{k_i}\}$. We conclude the following:

\begin{corollary} In the Broadcast CONGEST model, there is a deterministic algorithm that solves the all-pairs reachability problem in one round when the diameter of the underlying network is $1$. \end{corollary}

We also note that the time complexity of the internal computation of each vertex is $O(n^2)$.

\section{APSP Approximation for Networks of Diameter 1} \label{sec:apsp}

In the previous section, we showed that it is possible to solve the all-pairs reachability problem in one round for networks of diameter 1. Here we show that it is actually possible to compute an approximation to the distance between all pairs of vertices in two rounds using messages of at most $\lceil \log_2{n} \rceil$ bits (where $n$ is the number of vertices in the network).

Let $G = (V,E)$ be a directed graph on $n$ vertices and underlying diameter $1$. For every non-negative integer $i < n$, we let $A(i)$ be the set of all the vertices $u \in V$ whose out-degree is greater than $i$ and that have some in-going neighbor whose out-degree is at most $i$, that is, $A(i) = \{u \in V \mid (d_{out}(u) > i) \text{ \,and\, } (\exists w \in V \text{ s.t. } (w,u) \in E \text{ \,and\, } d_{out}(w) \le i)\}$. We also set $M(i)$ to be $\perp$ if $A(i)= \varnothing$ and $\max\,\{d_{out}(v) \mid v \in A(i)\}$ otherwise.

Next, we define for each vertex $x \in V$ a sequence $f_0[x],...,f_n[x]$ of elements by first setting $f_0[x]$ to be $\max\,\{d_{out}(v) \mid v \in \{x\}\cup N_{out}(x)\}$, and then for each $k \in \{1,...,n\}$, we set $f_k[x]$ to be $\perp$ if $f_{k-1}[x] = \;\perp$ and to $M(f_{k-1}[x])$ otherwise. We first prove some basic properties.

\begin{claim} For every $x \in V$ and $i \in \{0,...,n-1\}$, if $f_{i+1}[x] \ne \;\perp$ then $f_{i}[x] \ne \;\perp$ and $f_{i}[x] < f_{i+1}[x]$.\label{claim:basic_p}\end{claim}

\begin{proof} Let $x \in V$ and $i \in \{0,...,n-1\}$ be such that $f_{i+1}[x] \ne \;\perp$. By definition, we must have $f_{i}[x] \ne \;\perp$ and $A(f_{i}[x]) \ne \varnothing$, that is, $f_{i+1}[x]$ must be equal to the maximum out-degree of the vertices in $A(f_{i}[x])$. As, by definition, $A(f_{i}[x])$ contains only vertices whose out-degree is greater than $f_{i}[x]$, it must be that $f_{i+1}[x] > f_{i}[x]$.\end{proof}

Note that the above claim implies that if $f_i[x] \ne \;\perp$ holds for some $x \in V$ and $i \in \{1,...,n\}$, then we must have $f_j[x] \ne \;\perp$ for every $j \in \{0,...,i\}$ and $f_0[x] < ... < f_i[x]$. In particular, we must have $f_n[x] = \;\perp$ for every $x \in V$ (as otherwise, we would get that $f_n[x] \ge n$ which is impossible as the maximum possible out-degree is $n-1$).

In the next two lemmas, we show how the defined sequences can be used to estimate the distances between the vertices in the graph.

\begin{lemma} For every two different vertices $x$ and $y$ in $V$, if $y$ is reachable from $x$ in $G$, then there exists an index $i \in \{0,...,n-1\}$ such that $f_i[x] \ne \;\perp$ and $f_i[x] \ge d_{out}(y)$. Moreover, if $i$ is the minimal index for which this property holds, then the distance from $x$ to $y$ in $G$ is at least $i+1$. \label{lemma:est1}\end{lemma}

\begin{proof} Let $x$ and $y$ be two different vertices in $V$ such that $y$ is reachable from $x$, and let $\pi$ be some shortest path from $x$ to $y$ in $G$.

Let $i$ be the highest index in $\{0,...,n-1\}$ for which $f_i[x] \ne \; \perp$ (such an index must exist as $f_0[x] \ne \; \perp$ always holds). We claim that $f_i[x] \ge d_{out}(y)$. Indeed, assume towards a contradiction that this is not the case, that is, assume that $d_{out}(y) > f_i[x]$. As, in addition, we have $d_{out}(x) \le f_0[x] \le f_i[x]$, it must be that $\pi$ contains an edge $(u,v)$ such that $d_{out}(u) \le f_i[x]$ and $d_{out}(v) > f_i[x]$. It follows, by definition, that $v \in A(f_i[x])$ and so that $A(f_i[x]) \ne \varnothing$ which implies that $f_{i+1}[x] \ne \;\perp$. By the maximality of $i$, this is possible only if $i+1 = n$, that is, we must have $f_n[x] \ne \; \perp$ which is impossible. We conclude that there must be an index $i$ in $\{0,...,n-1\}$ for which $f_i[x] \ne \;\perp$ and $f_i[x] \ge d_{out}(y)$.

Now, let $j$ be the minimal index in $\{0,...,n-1\}$ with this property. We have to show that $\pi$ contains at least $j+1$ edges. As $x \ne y$, there must be a vertex $z \ne x$ which is adjacent to $x$ in $\pi$. By the definition of $f_0[x]$, we have $d_{out}(x) \le f_0[x]$ and $d_{out}(z) \le f_0[x]$. In other words, $\pi$ contains at least two different vertices whose out-degree is at most $f_0[x]$. We will show that $\pi$ must also contain at least $j$ different vertices whose out-degree is greater than $f_0[x]$. This will imply that $\pi$ contains at least $2 + j$ different vertices, and so at least $j+1$ edges.

First, note that $f_k[x] \ne \;\perp$ holds for every $k \in \{0,...,j\}$ (as $f_j[x] \ne \; \perp$) and so from the minimality of $j$ it must be that $f_k[x] < d_{out}(y)$ holds for every $k \in \{0,...,j-1\}$. We claim that for every such $k$ there must be a vertex $v_k$ in $\pi$ such that $f_k[x] < d_{out}(v_k) \le f_{k+1}[x]$. Indeed, let $k \in \{0,...,j-1\}$. As $d_{out}(x) \le f_0[x] \le f_k[x]$ and $f_k[x] < d_{out}(y)$, it must be that $\pi$ contains an edge $(u_k,v_k)$ such that $d_{out}(u_k) \le f_k[x]$ and $d_{out}(v_k) > f_k[x]$. By definition, we have $v_k \in A(f_k[x])$ and so $d_{out}(v_k) \le f_{k+1}[x]$, that is, $\pi$ contains a vertex $v_k$ such that $f_{k}[x] < d_{out}(v_k) \le f_{k+1}[x]$. It follows that $\pi$ must contain $j$ vertices $v_0,...,v_{j-1}$ such that $f_0[x] < d_{out}(v_0) < \dots < d_{out}(v_{j-1})$, that is, at least $j$ different vertices whose out-degree is greater than $f_0[x]$. \end{proof}

\begin{lemma} For every two vertices $x$ and $y$ in $V$ and every index $i \in \{0,...,n-1\}$, if $f_i[x] \ne \;\perp$ and $f_i[x] \ge d_{out}(y)$, then $G$ contains a path from $x$ to $y$ of length at most $3(i+1)$.\label{lemma:est2}\end{lemma}

We first prove the following auxiliary claim:

\begin{claim} For every two vertices $x$ and $y$ in $V$, if $d_{out}(x) \ge d_{out}(y)$ then $G$ contains a path from $x$ to $y$ of length at most $2$.\label{claim:dis2}\end{claim}

\begin{proof} Let $x$ and $y$ be two vertices in $V$ such that $d_{out}(x) \ge d_{out}(y)$ and assume towards a contradiction that the claim does not hold. We must have $x \ne y$ and $N_{in}(y) \cap \left(\{x\} \cup N_{out}(x)\right) = \varnothing$ as otherwise the distance from $x$ to $y$ would be at most $2$. Since the underlying diameter of $G$ is $1$, we get that $\{x\} \cup N_{out}(x) \subseteq N_{out}(y)$, and so that $d_{out}(y) = |N_{out}(y)| \ge |\{x\} \cup N_{out}(x)| > |N_{out}(x)| = d_{out}(x)$ which is a contradiction.\end{proof}

Now, we prove Lemma \ref{lemma:est2}:

\begin{proof} Let $x$ and $y$ be two vertices in $V$. Let $i \in \{0,...,n-1\}$ be such that $f_i[x] \ne \;\perp$ and $f_i[x] \ge d_{out}(y)$. We first show by induction that for every $k \in \{0,...,i\}$ there is a vertex $v_k \in V$ of out-degree $f_k[x]$ whose distance from $x$ is at most $1+3k$.

The base case ($k = 0$) holds as, by the definition of $f_0[x]$, there must be a vertex $v_0 \in \{x\} \cup N_{out}(x)$ such that $d_{out}(v_0) = f_0[x]$, that is, there must be a vertex $v_0 \in V$ of out-degree $f_0[x]$ whose distance from $x$ is at most $1$. Assume now that the claim holds for some $k \in \{0,...,i-1\}$ and prove it for $k+1$. By the induction hypothesis, there must be a vertex $v_k \in V$ whose out-degree is $f_k[x]$ and whose distance from $x$ is at most $1+3k$. We must have $f_{k+1}[x] \ne \; \perp$ (as $k+1\le i$) and so, by definition, there must be some edge $(u,v_{k+1}) \in E$ such that $d_{out}(u) \le f_k[x]$ and $d_{out}(v_{k+1}) = f_{k+1}[x]$. It follows that $d_{out}(u) \le d_{out}(v_k)$ and so from Claim \ref{claim:dis2} the distance from $v_k$ to $u$ is at most $2$. We get that the distance from $v_k$ to $v_{k+1}$ is at most $3$, and so the distance from $x$ to $v_{k+1}$ is at most $(1+3k)+3 = 1+3(k+1)$.

We conclude from the above that there must be some vertex $v_i \in V$ of out-degree $f_i[x]$ whose distance from $x$ is at most $1+3i$. As $f_i[x] \ge d_{out}(y)$ and $d_{out}(v_i) = f_i[x]$, we get that $d_{out}(v_i) \ge d_{out}(y)$. It follows, by Claim \ref{claim:dis2}, that the distance from $v_i$ to $y$ is at most $2$, and so the distance from $x$ to $y$ is at most $(1+3i)+2 = 3(i+1)$.\end{proof}

Lemmas \ref{lemma:est1} and \ref{lemma:est2} give us a way to estimate the distance from every vertex $x \in V$ to every other vertex $y \in V$ by just knowing $d_{out}(y)$ and the sequence $f_0[x],...,f_n[x]$. To see this, note first that these lemmas imply that $y$ is reachable from $x$ if and only if there exists an index $i \in \{0,...,n-1\}$ such that $f_i[x] \ne \;\perp$ and $f_i[x] \ge d_{out}(y)$. Furthermore, these lemmas imply that if $i$ is the minimal index with this property then $d(x,y,G) \le 3(i+1)\le 3\cdot d(x,y,G)$. Therefore, to compute an estimation $\hat{d}(x,y)$ for the distance $d(x,y,G)$, all we need to do is to find the minimal index (if any) $i \in \{0,...,n-1\}$ such that $f_i[x] \ne \;\perp$ and $f_i[x] \ge d_{out}(y)$, and then set $\hat{d}(x,y)$ to be $\infty$ if no such an index exists or to $3(i+1)$ otherwise.

Our next goal is to show that in two communication rounds each vertex $v \in V$ can learn the values of $d_{out}(x)$ and $f_0[x],...,f_n[x]$ of each vertex $x \in V$. To this end, we describe a slightly different but equivalent way to compute the sequence $f_0[x],...,f_n[x]$.

For every vertex $v \in V$ with positive out-degree, we let $m_v$ be the maximum out-degree of its out-going neighbors, that is, $m_v = \max\,\{d_{out}(u) \mid u \in N_{out}(v)\}$. For each $i \in \{0,...,n-1\}$, we set $A'(i)$ to be $\{m_v \mid v \in V \text{ \,and\, } m_v > i \text{ \,and\, } 0 < d_{out}(v) \le i\}$, and then we set $M'(i)$ to be $\perp$ if $A'(i) = \varnothing$ and $\max A'(i)$ otherwise.

\begin{claim}For every $i \in \{0,...,n-1\}$, we have $M(i) = M'(i)$.\end{claim}

\begin{proof} Let $i \in \{0,...,n-1\}$. We consider separately the cases $M(i) = \;\perp$ and $M(i) \ne \;\perp$. For the first case, note that, by definition, we must have $A(i) = \varnothing$, that is, there cannot be any vertex of out-degree at most $i$ that has an out-going neighbor of out-degree greater than $i$. In other words, we must have $m_v \le i$ for every $v \in V$ such that $0 < d_{out}(v) \le i$, and so $A'(i) = \varnothing$ which implies that $M'(i) = \;\perp$.

In the second case, we have $M(i) \ne \;\perp$ and so, by definition, there must be some vertex $u \in V$ of out-degree at most $i$ that has an out-going neighbor of out-degree $M(i) > i$. Moreover, $M(i)$ is the largest possible out-degree of any vertex in the out-going neighborhood of any vertex of out-degree at most $i$. This means that $m_u = M(i)$ and $m_v \le M(i)$ for every $v \in V$ such that $0 < d_{out}(v) \le i$, and so $M'(i) = M(i)$.\end{proof}

Now, given a vertex $x \in V$, it is easy to see that $f_0[x]$ is equal to $0$ if $d_{out}(x) = 0$ and to $\max\left(d_{out}(x), m_x\right)$ otherwise. Moreover, for every $i \in \{1,...,n\}$ we have, by the above claim, $f_{i}[x] = M'(f_{i-1}[x])$ if $f_{i-1}[x] \ne \;\perp$ and $f_{i}[x] = \;\perp$ otherwise. This observation gives rise to the following algorithm. Each vertex first broadcasts its out-degree to all of its neighbors in the underlying network. In the next round, each vertex with positive out-degree finds the maximal out-degree of its outgoing neighbors and broadcasts this value. By using the received information, each vertex can compute for every $x \in V$ and $i \in \{0,...,n\}$ the value of $f_i[x]$ and so compute a $3$-approximation to the distance between every pair of vertices. We conclude the following:

\begin{corollary} In the Broadcast CONGEST model, there is a deterministic algorithm that solves the 3-approximate APSP problem (for unweighted directed graphs) in two rounds when the diameter of the underlying network is 1.\end{corollary}

\section{Lower Bounds for Networks of Diameter 2} \label{sec:lower_bounds}

In this section we prove lower bounds for the single-source reachability problem and the closely related single-source shortest path problem for unweighted directed graphs in the Broadcast CONGEST model that hold even when the underlying network has diameter 2.

\subsection{The Single-Source Reachability Problem} \label{sec:lower_re}

We start this section by describing a family (parameterized by two positive integers $k$ and $q$) of directed graphs with underlying diameter at most $2$ which we denote by $F_{k,q}$. This family will be used later on to prove the required lower bound. 

\paragraph{The Family $\boldsymbol{F_{k,q}}$.} For two positive integers $k,q \in \mathbb{Z}$ and a $k$-bit string $\sigma \in \{0,1\}^k$, we define the directed graph $G(k,q,\sigma)$ to be the graph that consists of: \begin{itemize}\item[--] $k$ vertex-disjoint directed paths $P_1, ..., P_k$ with $q$ vertices each (that is, $P_i = (V_i,E_i)$ where $V_i = \{v_1^i,...,v_{q}^i\}$ and $E_i = \{(v_j^i, v_{j+1}^i) \mid j \in \{1,...,q-1\}\}$ for every $i \in \{1,...,k\}$).
\item[--] A source vertex $s$ that has an outgoing edge to $v_1^i$ (the first vertex of $P_i$) if the $i$-th bit of $\sigma$ is $1$, for every $i \in \{1,...,k\}$.
\item[--] A sink vertex $u$ to which $s$ and every vertex in $P_1, ..., P_k$ has an outgoing edge.\end{itemize}
In other words, the vertex set of the graph $G(k,q,\sigma)$ is $V_1 \cup ...\cup V_k \cup \{s,u\}$ and its edge set is $E_1 \cup ...\cup E_k \cup \{(s,v_{1}^i) \mid i \in \{1,...,k\} \text{ and } \sigma(i) = 1\} \cup \{(x,u) \mid x \in \{s\}\cup V_1 \cup ...\cup V_k\}$ (see Figure \ref{fig:graph} for an illustration). For two positive integers $k$ and $q$, we define the family $F_{k,q}$ to be the set $\{G(k,q,\sigma) \mid \sigma \in \{0,1\}^k\}$.

\begin{figure}[t]
	\centering
	\includegraphics[trim={0.1cm 1.2cm 1.0cm 0.23cm},clip]{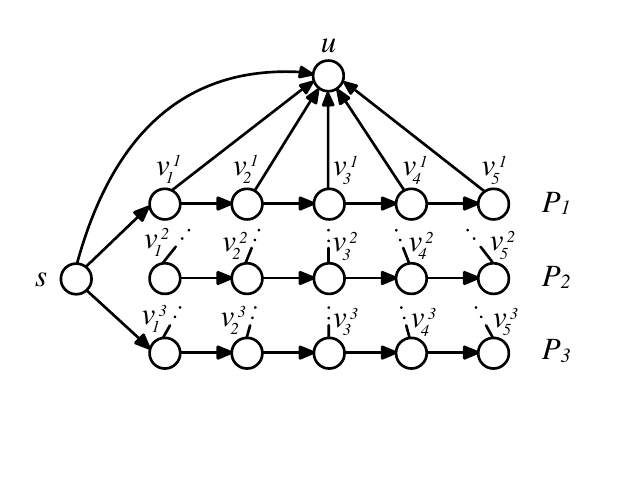}
	\caption{An illustration of the graph $G(k,q,\sigma)$ for $k = 3$, $q = 5$ and $\sigma = 101$. \label{fig:graph}} \end{figure}

Our next goal is to show that any distributed algorithm that solves the single-source reachability problem for all the graphs in $F_{k,q}$ requires a significant number of rounds. We start with the following lemma:

\begin{lemma} Let $k$ and $q$ be two positive integers. Let $G \in F_{k,q}$ and let $\varphi$ be some legal assignment of identifiers to its vertices. Let $A$ be some deterministic distributed algorithm (in the Broadcast CONGEST model) that solves the single-source reachability problem on the instance $(G,\varphi,s)$ using at most $t$ rounds (for some non-negative integer $t < q$). For each $i \in \{1,...,k\}$, the output of the vertex $v_{q}^i$ by the end of the last round is just a function of the initial input of the vertices $v_{q-t}^i, ..., v_{q}^i$ and the sequence of messages that $v_{q}^i$ received from $u$. \label{lem:limited_dep}\end{lemma}

\begin{proof} Let $i \in \{1,...,k\}$. We will show by induction on $0 \le j \le t$ that by the end of the $j$-th round the state of each vertex $v \in \{v_{q-t+j}^i,...,v_q^{i}\}$ is just a function of the initial input of the vertices in its ball of radius $j$ (in the underlying graph of $P_i$) and the sequence of messages that it received from $u$ up to this round.

The base case ($j = 0$) clearly holds as the state of each vertex in $\{v_{q-t}^i,...,v_{q}^i\}$ by the end of round $0$ can depend only on its initial input. Assume now that the claim holds for some $0 \le j < t$ and prove it for $j+1$. Let $r \in \{q-t+j+1,...,q\}$. The state of $v_r^i$ by the end of the ($j+1$)-th round is a function of its state at the end of the previous round and the messages that it received from its neighbors in the underlying network (which are $u$, $v_{r-1}^i$ and possibly $v_{r+1}^i$).

The messages that $v_r^i$ has received from its neighbors in $V_i$ are, by the induction hypothesis, functions of the inputs of the vertices in the balls of radius $j$ (in $P_i$) around these neighbors and the sequence of messages that they received from $u$ (up to round $j$). As $u$ broadcasts the same message to all the vertices in each round, we get that these messages are just a function of the inputs of the vertices in the ball of radius $j+1$ around $v_r^i$ (in $P_i$) and the sequence of messages that $v_r^i$ received from $u$ (up to round $j$). As the previous state of $v_r^i$ is, by the induction hypothesis, also a function of the initial inputs of the vertices in its ball of radius $j$ (in $P_i$) and the sequence of message that it received from $u$, the claim follows. \end{proof}

\begin{lemma} Let $k$ and $q$ be two positive integers and let $\varphi$ be some legal assignment of identifiers to $V_1\cup...\cup V_k\cup\{s,u\}$. For every deterministic algorithm $A$ (in the Broadcast CONGEST model), if $A$ solves the single-source reachability problem on all the instances in $\{(G,\varphi,s) \mid G \in F_{k,q}\}$ and uses messages of size at most $B$ bits (for some $B \ge 1$), then $A$ requires at least $\min\{q - 1, k/(2B)\}$ rounds. \label{lem:lower_bound}\end{lemma}

\begin{proof} Let $A$ be some deterministic algorithm that satisfies the requirements of the lemma and let $t$ be its running time. We can assume that $t \le q-2$ as otherwise there is nothing to show.

For each $G \in F_{k,q}$, we let $\text{out}(G)$ be the sequence $(\text{out}(v_q^1,G),...,\text{out}(v_q^k,G))$ where $\text{out}(v_q^i,G)$ is the output of $v_q^i$ when $A$ is invoked on $(G,\varphi,s)$, for every $i \in \{1,...,k\}$. Lemma \ref{lem:limited_dep} implies that for each $G \in F_{k,q}$ the value of $\text{out}(G)$ is just a function of the initial inputs in $(G,\varphi,s)$ of the vertices $\bigcup_{i=1}^k\{v_{q-t}^i,...,v_{q}^i\}$ and the sequence of messages that $u$ had broadcast. Since we assumed that $t \le q-2$, the initial inputs of these vertices is the same in all $\{(G,\varphi,s) \mid G \in F_{k,q}\}$, and so we can have $\text{out}(G) \ne \text{out}(G')$ for two graphs $G$ and $G'$ in $F_{k,q}$ only if the sequence of messages that $u$ had broadcast in the corresponding invocations was different.

In each round, $u$ may send a message that contains at most $B$ bits, that is, a message with $0$ bits, or with $1$ bit and so on. Therefore, there are $1+2+...+2^B \le 2^{2B}$ different messages that $u$ may send in each round. It follows that there are at most $2^{2Bt}$ possible sequences and so we get that $|\{\text{out}(G) \mid G \in F_{k,q}\}| \le 2^{2Bt}$. Note also that $|\{\text{out}(G) \mid G \in F_{k,q}\}| = 2^{k}$ as for each $G \in F_{k,q}$ the output should be different. We conclude that $2^{k} \le 2^{2Bt}$ and so $t \ge k/2B$.\end{proof}

\begin{corollary} In the Broadcast CONGEST model, there is no deterministic algorithm that solves the single-source reachability problem in $o(\sqrt{n/\log{n}}\,)$ rounds even when the diameter of the underlying network is always $2$.\end{corollary}

\begin{proof} Assume towards a contradiction that there exists a deterministic algorithm $A$ that solves the above problem in $T(n) = o(\sqrt{n/\log{n}}\,)$ rounds. As $A$ works in the CONGEST model, there must be some constant $c \ge 1$ such that the number of bits in any message that the algorithm may send (when it is invoked on inputs of size $n > 1$) is at most $c\cdot\log{n}$.

Since $T(n) = o(\sqrt{n/\log{n}}\,)$, there must be some integer $n_0 \ge 16$ for which $T(n) \le \frac{1}{10c}\sqrt{n/\log{n}}$ holds for every $n > n_0$. Choose an integer $m > n_0$ such that both $k = \sqrt{m\log{m}}$ and $q = \sqrt{m/\log{m}}$ are positive integers. Lemma \ref{lem:lower_bound} implies that there must be some $G \in F_{k,q}$ and some assignment of identifiers $\varphi$ to $V(G)$ such that invoking the algorithm on $(G,\varphi,s)$ requires at least $\min\{\frac{1}{2}q,\frac{1}{4c}\frac{k}{\log{m}}\} = \min\{\frac{1}{2}\sqrt{\frac{m}{\log{m}}},\frac{1}{4c}\sqrt{\frac{m}{\log{m}}}\} = \frac{1}{4c}\sqrt{\frac{m}{\log{m}}}$ rounds. But, we also have $|V(G)| > m > n_0$ and so the algorithm must take at most $\frac{1}{5c}\sqrt{m/\log{m}}$ rounds on $(G,\varphi,s)$, a contradiction.\end{proof}

In the next lemma, we show that the same lower bound holds for distributed randomized algorithms as well.

\begin{lemma} Let $k$ and $q$ be two positive integers and let $\varphi$ be some legal assignment of identifiers to $V_1\cup...\cup V_k\cup \{s,u\}$. For every randomized algorithm $A$ (in the Broadcast CONGEST model), if $A$ correctly solves the SSR problem on each instance in $\{(G,\varphi,s)\mid G \in F_{k,q}\}$ with probability $> 1/2$ and uses messages of size at most $B$ bits (for some $B\ge1$), then $A$ requires at least $\min\{q - 1,\, (k-1)/(2B)\}$ rounds.\label{lem:lower_bound_rand}\end{lemma}

\begin{proof} Clearly, it is sufficient to show that this lower bound holds in a model that generates a public random string first, announces it to every vertex in the graph and then every vertex proceeds deterministically as usual. Let $A$ be a randomized algorithm that works in the above model and solves the SSR problem on every instance in $F = \{(G,\varphi, s) \mid G \in F_{k,q}\}$ with probability $> 1/2$. We can assume that its running time $t$ is at most $q-2$. As in the proof of Lemma \ref{lem:lower_bound}, we can show that, for every fixed random string $r$, the algorithm (given that string) can succeed on at most $2^{2Bt}$ of the instances in $F$. 

For each graph $G$ in $F_{k,q}$, let $R_G$ be the event that the algorithm fails on $(G,\varphi,s)$. Note that we must have $\sum_{G\in F_{k,q}}P(R_G) \ge |F_{k,q}| - 2^{2Bt}$. By assumption, we must also have $0.5|F_{k,q}| > \sum_{G\in F_{k,q}}P(R_G)$ and so $0.5|F_{k,q}| > |F_{k,q}| - 2^{2Bt}$ which implies that $t > (k-1)/2B$.\end{proof}

\subsection{The Single-Source Shortest Path Problem} \label{sec:lower_sssp}

The result of the previous section already gives a lower bound of $\Omega(\sqrt{n/\log{n}}\,)$ for the directed SSSP problem (or even for the approximate version of it) as, by definition, any algorithm that solves this problem must also solve the SSR problem.

In this section we show a slightly stronger lower bound of $\Omega(\sqrt{n}\,)$ for this problem which holds even when the diameter of the underlying network is 2 and even when all the vertices in the input graph are guaranteed to be reachable from the given source. As in the previous section, we start by describing a family $J_k$ of unweighted directed graphs with underlying diameter 2 which will be used to prove the lower bound.

\paragraph{The Family $\boldsymbol{J_{k}}$.} For a positive integer $k$ and a sequence $\sigma$ of $k$ numbers from $\{1,...,k\}$, we define the directed graph $G(k,\sigma)$ to be the graph that consists of: \begin{itemize}\item[--] $k$ vertex-disjoint directed paths $P_1, ..., P_k$ where each $P_i = (V_i,E_i)$ contains $\sigma(i) + k$ vertices. For each path $P_i$, we denote by $u^i$ its first vertex and by $v_1^i,...,v_{k}^i$ its last $k$ vertices.
\item[--] A source vertex $s$ that has an outgoing edge to the first vertex of every path in $\{P_1, ..., P_k\}$.
\item[--] A sink vertex $u$ to which $s$ and every vertex in $P_1, ..., P_k$ has an outgoing edge.\end{itemize}
In other words, the vertex set of the graph $G(k,\sigma)$ is $V_1 \cup ...\cup V_k \cup \{s,u\}$ and its edge set is $E_1 \cup ...\cup E_k \cup \{(s,u^i) \mid i \in \{1,...,k\}\} \cup \{(x,u) \mid x \in \{s\}\cup V_1 \cup ...\cup V_k\}$ (see Figure \ref{fig:graph1} for an illustration). For a positive integer $k$, we define the family $J_{k}$ to be the set $\{G(k,\sigma) \mid \sigma \in \{1,...,k\}^k\}$. 

\begin{figure}[t]
	\centering
	\includegraphics[trim={0.2cm 0.8cm 0.6cm 0.4cm},clip]{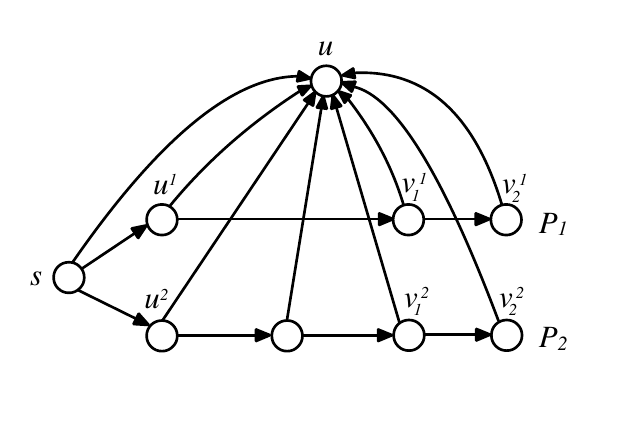}
	\caption{An illustration of the graph $G(k,\sigma)$ for $k = 2$ and $\sigma = (1,2)$. \label{fig:graph1}} 
\end{figure}

We say that a collection of assignments $\{\varphi_G \mid G \in J_{k}\}$ is a consistent set of assignments for the family $J_{k}$, if $\varphi_{G}$ is a legal assignment of identifiers to $V(G)$ and $\varphi_{G}(x) = \varphi_{G'}(x)$, for every $G,G' \in J_k$ and $x \in \{u\} \cup (\bigcup_{i=1}^k\{v_{1}^i,...,v_{k}^i\})$.

\begin{lemma} Let $k > 1$ be some integer and let $\{\varphi_G \mid G \in J_{k}\}$ be some consistent set of assignments for $J_k$. For every deterministic algorithm $A$ (in the Broadcast CONGEST model), if $A$ solves the SSSP problem on all the instances in $\{(G,\varphi_G,s) \mid G \in J_k\}$, then $A$ requires $\Omega(k)$ rounds.\label{lem:lower_bound_sssp}\end{lemma}

\begin{proof} Let $A$ be some deterministic algorithm that satisfies the requirements of the lemma and let $t$ be its running time. We can assume that $t \le k-2$ as otherwise there is nothing to show. As $A$ works in the CONGEST model, there must be some constant $c \ge 1$ such that the number of bits in any message that the algorithm may send on any input of size $n > 1$ is at most $c\cdot \log{n}$.

Consider invoking the algorithm $A$ on the instance $(G,\varphi_G,s)$ for some $G \in J_k$. In each round, $u$ may send one message containing at most $c\cdot \log(|V(G)|) \le 4c \cdot \log(k)$ bits to all the vertices in the graph. Given that, it is easy to see (by a proof similar to Lemma \ref{lem:limited_dep}) that the output of every $v_k^i$ is just a function of its initial input, the initial input of at most $t \le k-2$ vertices that precede it in the path $P_i$, and the sequence of messages that $u$ had broadcast.

For each $G \in J_k$, we let $\text{out}(G)$ be the sequence $(\text{out}(v_k^1,G),...,\text{out}(v_k^k,G))$ where $\text{out}(v_k^i,G)$ is the output of $v_k^i$ when $A$ is invoked on $(G,\varphi_G,s)$, for every $i \in \{1,...,k\}$. By the observation above, for each $i \in \{1,...,k\}$ the vertices in $P_i$ whose initial input may affect the output of $v_k^i$ are just $v_{2}^i,...,v_{k}^i$. Since the initial input of each of these vertices is the same in each of the instances in $\{(G,\varphi_G,s) \mid G \in J_k\}$, we get that $\text{out}(G) \ne \text{out}(G')$ can hold for some graphs $G$ and $G'$ in $J_k$ only if the sequence of messages that $u$ had broadcast in the corresponding invocations was different.

Straightforward calculations show that the number of such sequences is at most $(2^{8c \cdot \log(k)})^t$ $= k^{8c\cdot t}$, and so $|\{\text{out}(G) \mid G \in J_k\}| \le k^{8c\cdot t}$. Since we assumed that the algorithm is correct, we must have $\text{out}(G) \ne \text{out}(G')$ for every two different graphs $G$ and $G'$ in $J_k$ (as, by construction, we cannot have $d(s,v_k^i,G) = d(s,v_k^i,G')$ for every $i \in \{1,...,k\}$), and so $|\{\text{out}(G) \mid G \in J_k\}| = |J_k| = k^k$. We conclude that $k^k \le k^{8c\cdot t}$ and so $t \ge k / (8c)$.\end{proof}

\begin{corollary} In the Broadcast CONGEST model, there is no deterministic algorithm that solves the SSSP problem in $o(\sqrt{n}\,)$ rounds even when the diameter of the underlying network is always $2$. \end{corollary}

\begin{proof} Assume towards a contradiction that there exists a deterministic algorithm $A$ that solves the above problem in $T(n) = o(\sqrt{n}\,)$ rounds. As before, we can assume that there is a constant $c \ge 1$ such that the number of bits in any message that $A$ may send on any input of size $n > 1$ is at most $c\cdot\log{n}$.

Since $T(n) = o(\sqrt{n}\,)$, there must be some positive integer $n_0$ for which $T(n) \le \frac{1}{18c}\sqrt{n}$ holds for every $n > n_0$. Let $k > 1$ be an integer such that $k^2 > n_0$. The proof of Lemma \ref{lem:lower_bound_sssp} implies that there must be some $G \in J_{k}$ and some assignment of identifiers $\varphi$ to $V(G)$ such that invoking the algorithm on $(G,\varphi,s)$ requires at least $k/(8c)$ rounds. But, $|V(G)| > k^2 > n_0$ and so the algorithm must take at most $\frac{1}{18c}\sqrt{|V(G)|} \le \frac{1}{9c}k$ rounds on $(G,\varphi,s)$, a contradiction.\end{proof}

By a proof similar to that of the previous section, it is possible to show that the same lower bound holds for randomized distributed algorithms as well.

\subsection{A Note About the Topology of the Networks} \label{sec:topology}

The graphs that were constructed in the previous two sections are very simple in terms of their structure. In fact, by a slight modification, it is possible to simplify their structure even further, and thus, to extend the lower bounds to a wider class of graph families. In this section we illustrate this modification on the family $F_{k,q}$ (a similar modification also applies to $J_k$).

For two positive integers $k,q \in \mathbb{Z}$ and a $k$-bit string $\sigma \in \{0,1\}^k$, we define the directed graph $G'(k,q,\sigma)$ to be the graph that is obtained from $G(k,q,\sigma)$ by removing the vertex $s$ (together with its adjacent edges) and reversing the direction of the edge $(v_i^1,u)$ for every $i \in \{1,...,k\}$ such that $\sigma(i) = 1$ (see Figure \ref{fig:graph2} for an illustration).

\begin{figure}[t]
	\centering
	\includegraphics[trim={0cm 0.1cm 0cm 0.2cm},clip]{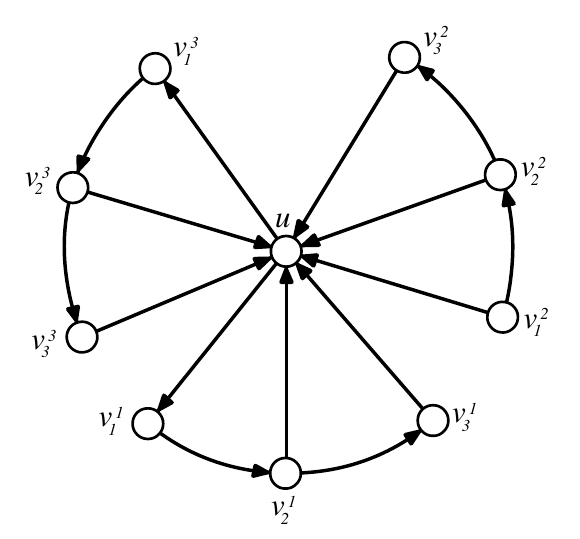}
	\caption{An illustration of the graph $G'(k,q,\sigma)$ for $k = 3$, $q = 3$ and $\sigma = 101$. \label{fig:graph2}} \end{figure}

The above modification simplifies the structure of the graphs in $F_{k,q}$ without affecting the lower bound. More precisely, given some legal assignment of identifiers $\varphi$ to $V_1 \cup ... \cup V_k \cup \{u\}$, it is straightforward to adapt the proof of Lemma \ref{lem:lower_bound} (and similarly of Lemma \ref{lem:lower_bound_rand}) to show that any algorithm (in the Broadcast CONGEST model) that uses messages of size at most $B$ and solves the SSR problem on all the instances in $\{(G,\varphi,u) \mid G \in F'_{k,q}\}$, where $F'_{k,q} = \{G'(k,q,\sigma) \mid \sigma \in \{0,1\}^k\}$, requires $\Omega(\min\{q,k/B\})$ rounds.

\bibliographystyle{plain}
\bibliography{ref}

\begin{thebibliography}{10}

\bibitem{AgaU18}
Udit Agarwal, Vijaya Ramachandran, Valerie King, and Matteo Pontecorvi.
\newblock A deterministic distributed algorithm for exact weighted all-pairs
  shortest paths in $\tilde{O}(n^{3/2})$ rounds.
\newblock In {\em Proceedings of the 2018 ACM Symposium on Principles of
  Distributed Computing (PODC '18)}, pages 199--205, 2018.

\bibitem{Beck17}
Ruben Becker, Andreas Karrenbauer, Sebastian Krinninger, and Christoph Lenzen.
\newblock Near-optimal approximate shortest paths and transshipment in
  distributed and streaming models.
\newblock In {\em Proceedings of the 31st International Symposium on
  Distributed Computing (DISC '17)}, pages 7:1--7:16, 2017.

\bibitem{AD19}
Aaron Bernstein and Danupon Nanongkai.
\newblock Distributed exact weighted all-pairs shortest paths in near-linear
  time.
\newblock In {\em Proceedings of the 51st Annual ACM SIGACT Symposium on Theory
  of Computing (STOC '19)}, pages 334--342, 2019.

\bibitem{CensorH19}
Keren Censor-Hillel, Michal Dory, Janne~H. Korhonen, and Dean Leitersdorf.
\newblock Fast approximate shortest paths in the {Congested Clique}.
\newblock In {\em Proceedings of the 2019 ACM Symposium on Principles of
  Distributed Computing (PODC '19)}, pages 74--83, 2019.

\bibitem{CensorH15}
Keren Censor-Hillel, Petteri Kaski, Janne~H. Korhonen, Christoph Lenzen, Ami
  Paz, and Jukka Suomela.
\newblock Algebraic methods in the {Congested Clique}.
\newblock In {\em Proceedings of the 2015 ACM Symposium on Principles of
  Distributed Computing (PODC '15)}, pages 143--152, 2015.

\bibitem{DasSa11}
Atish Das~Sarma, Stephan Holzer, Liah Kor, Amos Korman, Danupon Nanongkai,
  Gopal Pandurangan, David Peleg, and Roger Wattenhofer.
\newblock Distributed verification and hardness of distributed approximation.
\newblock In {\em Proceedings of the Forty-third Annual ACM Symposium on Theory
  of Computing (STOC '11)}, pages 363--372, 2011.

\bibitem{Elkin17}
Michael Elkin.
\newblock Distributed exact shortest paths in sublinear time.
\newblock In {\em Proceedings of the 49th Annual ACM SIGACT Symposium on Theory
  of Computing (STOC '17)}, pages 757--770, 2017.

\bibitem{SFDN18}
Sebastian Forster and Danupon Nanongkai.
\newblock A faster distributed single-source shortest paths algorithm.
\newblock In {\em Proceedings of the 59th Annual IEEE Symposium on Foundations
  of Computer Science (FOCS '18)}, pages 686--697, 2018.

\bibitem{GhaLi18}
Mohsen Ghaffari and Jason Li.
\newblock Improved distributed algorithms for exact shortest paths.
\newblock In {\em Proceedings of the 50th Annual ACM SIGACT Symposium on Theory
  of Computing (STOC'18)}, pages 431--444, 2018.

\bibitem{GR15}
Mohsen Ghaffari and Rajan Udwani.
\newblock Brief announcement: Distributed single-source reachability.
\newblock In {\em Proceedings of the 2015 ACM Symposium on Principles of
  Distributed Computing (PODC '15)}, pages 163--165, 2015.

\bibitem{Hen16}
Monika Henzinger, Sebastian Krinninger, and Danupon Nanongkai.
\newblock A deterministic almost-tight distributed algorithm for approximating
  single-source shortest paths.
\newblock In {\em Proceedings of the Forty-eighth Annual ACM Symposium on
  Theory of Computing (STOC '16)}, pages 489--498, 2016.

\bibitem{Holz12}
Stephan Holzer and Roger Wattenhofer.
\newblock Optimal distributed all pairs shortest paths and applications.
\newblock In {\em Proceedings of the 2012 ACM Symposium on Principles of
  Distributed Computing (PODC '12)}, pages 355--364, 2012.

\bibitem{CDT17}
Chien-Chung Huang, Danupon Nanongkai, and Thatchaphol Saranurak.
\newblock Distributed exact weighted all-pairs shortest paths in
  $\tilde{O}(n^{5/4})$ rounds.
\newblock In {\em Proceedings of the 58th Annual IEEE Symposium on Foundations
  of Computer Science (FOCS '17)}, pages 168--179, 2017.

\bibitem{Arun19}
Arun Jambulapati, Yang~P. Liu, and Aaron Sidford.
\newblock Parallel reachability in almost linear work and square root depth.
\newblock In {\em Proceedings of the 60th Annual IEEE Symposium on Foundations
  of Computer Science (FOCS '19)}, 2019.

\bibitem{Len15}
Christoph Lenzen and Boaz Patt-Shamir.
\newblock Fast partial distance estimation and applications.
\newblock In {\em Proceedings of the 2015 ACM Symposium on Principles of
  Distributed Computing (PODC '15)}, pages 153--162, 2015.

\bibitem{Len13}
Christoph Lenzen and David Peleg.
\newblock Efficient distributed source detection with limited bandwidth.
\newblock In {\em Proceedings of the 2013 ACM Symposium on Principles of
  Distributed Computing (PODC '13)}, pages 375--382, 2013.

\bibitem{Nan14}
Danupon Nanongkai.
\newblock Distributed approximation algorithms for weighted shortest paths.
\newblock In {\em Proceedings of the Forty-sixth Annual ACM Symposium on Theory
  of Computing (STOC '14)}, pages 565--573, 2014.

\bibitem{Peleg00b}
David Peleg.
\newblock {\em Distributed Computing: A Locality-sensitive Approach}.
\newblock Society for Industrial and Applied Mathematics, 2000.

\end{thebibliography}

\end{document}